\documentclass[12pt]{amsart}  

\usepackage{bbold}
\usepackage{amscd} 
\usepackage{amsmath,amssymb}
\usepackage{mathtools}
\usepackage{mathrsfs}
\newtheorem{theorem}{Theorem}[section]
\newtheorem{lemma}[theorem]{Lemma}
 
\theoremstyle{definition}

\newtheorem{remark}[theorem]{Remark} 
\numberwithin{equation}{section}

\def\A{\mathbb A}
\def\B{\mathscr B}
\def\M{\mathbb M}

\def\uno{1}
\def\Sob{\mathfrak H}
\def\H{\mathfrak H}
\def\F{\mathfrak F}

\def\D{\text{\rm dom}}
\def\dom{\text{\rm dom}}
\def\ran{\text{\rm ran}}

\def\G{\mathscr G}
\def\RE{\mathbb R}
\def\CO{{\mathbb C}}

\def\ph*{\phi_\star}

\def\be{\begin{equation}}
\def\ee{\end{equation}}
\def\min{{\rm min}}
\def\max{{\rm max}}

\def\-{{\rm in}}
\def\+{{\rm ex}}

\def\op{\overset{\infty}{\underset{n= 0}{\oplus}}}
\def\x{\mathsf x}

\def\Th{\mathbb\Theta}

\def\G{\mathbb G}

\def\S{\mathfrak D}
\def\Ml{\M}
\def\Thl{\mathbb\Theta}

\newcommand{\closure}[2][3]{%
      {}\mkern#1mu\overline{\mkern-#1mu#2}}
\def\cH{\closure{H}}

\begin{document}

\centerline{{\bf ON THE RESOLVENT OF  H+A${\,}^{\!\!\!*}$+A}}
\vskip10pt
\centerline{Andrea Posilicano}\vskip10pt
\centerline{DiSAT, Sezione di Matematica, Universit\`a dell'Insubria, Como, Italy}
\centerline{\tt posilicano@uninsubria.it}
\vskip20pt\noindent
{\bf Abstract.} We present a much shorter and streamlined proof of an improved version of the results previously 
given in [A. Posilicano: On the Self-Adjointness of $H+A^{*}+A$. {\it Math. Phys. Anal. Geom.} {\bf 23} (2020)] concerning the self-adjoint realizations of formal QFT-like Hamiltonians of the kind $H+A^{*}+A$, where $H$ and $A$ play the role of the free field Hamiltonian and of the annihilation operator respectively. We give explicit representations of the resolvent and of the self-adjointness domain; the consequent Kre\u\i n-type resolvent formula leads to a characterization of these self-adjoint realizations as limit (with respect to convergence in norm resolvent sense)  of cutoff Hamiltonians of the kind $H+A^{*}_{n}+A_{n}-E_{n}$, the bounded operator $E_{n}$ playing the role of a renormalizing counter term. These abstract results apply to various concrete models in Quantum Field Theory. 

\begin{section}{Introduction}
The Nelson Hamiltonian is the (renormalized) self-adjoint operator corresponding to the quantization of the classical system 
\begin{align*}
\ddot\phi(t)=&(\Delta-\mu^{2})\phi(t)-g\sum_{i=1}^{N}\delta_{q_{i}(t)}\\
m\ddot q_{i}(t)=&-g[\nabla\phi (t)](q_{i}(t))\,,\quad i=1,\dots, N\,.
\end{align*}
It models the interaction in three-dimensional space (with coupling constant $g\in\RE$) between a scalar field with mass $\mu\ge0$ and $N$ point particles with mass $m>0$. \par
The singularity created by the presence of the moving Dirac's delta sources $\delta_{q_{i}(t)}$ prevents the evaluation of $\nabla\phi$ at the particle positions $q_{i}(t)$. Therefore, the classical equations are ill-defined and the system cannot be quantized as it stands;  
one regularizes the interaction by introducing an ultraviolet cutoff $\Lambda$, thus obtaining a well-defined, self-adjoint and  bounded-from-below  Hamiltonian $$H_{\Lambda}:=H+A_{\Lambda}^{*}+A_{\Lambda}\,.
$$ 
Here, $H$ is the positive self-adjoint operator (representing the free Hamiltonian) in the Hilbert space $\F=L^{2}(\RE^{3N})\otimes\F_{b}\simeq L^{2}(\RE^{3N};\F_{b})$ defined by
\be\label{free0}
H:=-\frac{\Delta}{2m}\otimes \uno+\uno\otimes {\rm d}\Gamma_{\! b}\big((-\Delta+\mu^{2})^{1/2}\big)\,,
\ee
where $\F_{b}:=\Gamma_{\! b}(L^{2}(\RE^{3}))\equiv \op L^{2}_{sym}(\RE^{3n})$ denotes the boson Fock space over $L^{2}(\RE^{3})$ and
${\rm d}\Gamma_{\! b}(L)$ denotes the bosonic second quantization of $L$, 
and 
\be\label{ann0}
(A_{\Lambda}\Psi)(\x):=\frac{1}2\,a(v^{\Lambda}_{\x})\Psi(\x)\,,\qquad v^{\Lambda}_{ \x}:=g
\sum_{i=1}^{N}(-\Delta+\mu^{2})^{-1/4}\delta^{\Lambda}_{x_{i}}\,,\quad \x\equiv(x_{1},\dots,x_{N})\,,
\ee
where $a(v)$ denotes the annihilation operator in $\F_{b}$ with test vector $v$, and $\delta_{x_{i}}^{\Lambda}$ is the regularization of Dirac's delta defined through the Fourier transform $\widehat{\delta_{x_{i}}^{\Lambda}}:=\chi_{\Lambda}\widehat{\delta_{x_{i}}}$, $\chi_{\Lambda}$ being the characteristic function of the ball of radius $\Lambda$. Thanks to the ultraviolet cutoff $\Lambda$, $A_{\Lambda}^{*}+A_{\Lambda}$ turns out to be $H$-small; hence,  by the Rellich-Kato theorem, $H_{\Lambda}$ is a well defined,  bounded-from-below, self-adjoint operator.\par 
Computing ${\mathscr E}_{\Lambda}$, the ground state energy  at zero total momentum of $H_{\Lambda}$,  one gets
$$
{\mathscr E}_{\Lambda}=-g^{2}N{\mathcal E}_{\Lambda}+O(g^{4})\,,\qquad {\mathcal E}_{\Lambda}\sim \log\Lambda,\  \Lambda\gg1\,.
$$
In the sixties, Edward Nelson  proved   that  there exists 
a  bounded-from-below  self-adjoint operator $\widehat H$ such that 
\be\label{Nelson}
\lim_{\Lambda\nearrow\infty}e^{-it(H_{\Lambda}-g^{2}N{\mathcal E}_{\Lambda})}\Psi
=e^{-it\widehat H}\Psi
\ee
for any time $t\in\RE$ and any $\Psi\in\F$ (see the seminal paper \cite{N} and its antecedent \cite{N0}). As it is well known, \eqref{Nelson} is equivalent to the strong resolvent convergence of $H_{\Lambda}-g^{2}N{\mathcal E}_{\Lambda}$ to $\widehat H$; however, no result neither regarding  the resolvent nor the self-adjointness domain of $\widehat H$ was known up to very recent times.\par A first result about the form-domain of $\widehat H$ was obtained in \cite{GW2} by Griesemer and  W\"unsch; shortly after,  Lampart and Schmidt (see \cite{LS}, also see \cite{L2}, \cite{L3}, \cite{Sch1}, \cite{Sch2} for successive analogous results concerning related models)  obtained an explicit representation of the self-adjointness domain and of the action of $\widehat H$ on it:
\be\label{DL}
\D(\widehat H)=\{\Psi\in\F:\Psi_{0}:=\Psi-G_{0}\Psi\in\D(H)\}
\ee
and
\be\label{AL}
\widehat H\Psi=\cH\Psi+A^{*}\Psi+\widehat A\Psi\,,\qquad \widehat A\Psi:=A\Psi_{0}-\widehat T\Psi\,.
\ee
Here, $G_{z}:=(AR_{\bar z})^{*}$ with $R_{z}$ the resolvent of $H$ at $z\in\varrho(H)$ and 
the bounded (with respect to the graph norm in $\D(H)$) operator $$A:\D(H)\to\F$$ is defined as in \eqref{ann0} with the Dirac delta $\delta_{x_{i}}$ in place of the regularized $\delta_{x_{i}}^{\Lambda}$, $\cH$ denotes the closure of the bounded operator $H:\D(H)\subseteq\F\to\D(H)^{*}$ 
and $\widehat T:\D(\widehat H)\subseteq\F\to\F$ is an explicit symmetric operator. The operator  $\cH+A^{*}$ turns out to be $\F$-valued whenever restricted to $\D(\widehat H)$ and so $\widehat H$ is $\F$-valued as well. \par
We refer to, e.g., \cite{Arai} for a systematic introduction to the mathematics of quantum fields; however, let us point out that, in our successive abstract setting, we do not need any specific property neither of the Fock space nor of the annihilation operator beside the assumptions that $A$ is $H$-bounded (with respect to a bounded-from-below self-adjoint $H$ in a Hilbert space $\F$) and has a  kernel which is dense in $\F$ (actually, a weaker hypothesis suffices, see Remark \ref{REM} below).\par
The representation \eqref{DL} much resembles the one for the domain of a self-adjoint extension of the symmetric restriction $H|\ker(A)$. In this case, the domain is of the kind 
 $$\{\Psi\in\F:\Psi_{0}:=\Psi-G_{0}\Phi\in\D(H),\ A\Psi_{0}=\Theta\Phi\},$$ 
 for some self-adjoint operator $\Theta:\dom(\Theta)\subseteq\F\to\F$ and the corresponding self-adjoint 
 extension $H_{\Theta}$ is characterized by the  Kre\u\i n-type resolvent formula  $$
 (-H_{\Theta}+z)^{-1}=(-H+z)^{-1}+G_{z}(\Theta+A(G_{0}-G_{z}))^{-1}G_{\bar z}^{*}\,,
 $$ 
see \cite{JFA} and \cite{O&M} (we refer to \cite{O&M} for the connections with classical von Neumann's scheme  
and with boundary triplets theory). However, there is a crucial difference between $\widehat H$ and $H_{\Theta}$: since  it turns out that, as is typical in quantum fields models, the domain of $\widehat H$ has trivial intersection with the domain of the free Hamiltonian, while $\ker(A)$ is dense, one gets
$$\dom(\widehat H)\cap\dom(H)=\{0\}\not=\ker(A)={\dom(H_{\Theta})\cap\dom(H)}\,.
$$
This shows that one cannot hope to obtain $\widehat H$ by a straightforward application of the theory of singular perturbations of self-adjoint operators (a.k.a. self-adjoint extensions of symmetric restrictions). \par In \cite{MPAG}, such an obstruction was circumvented by applying the restriction-extension procedure twice: at first one builds the self-adjoint extensions of $H|\ker(A)$ and then one gets $\widehat H$ as a particular self-adjoint extension of the restriction of a first-step extension  to $\ker{(1-A_{*})}$, where $A_{*}$ is a suitable left inverse of $G_{0}$; a Kre\u\i n-type resolvent formula for $\widehat H$ was then obtained by combining the resolvents of the two extensions.\par Here, we follow a much shorter and easy path: we apply the original scheme provided in \cite{JFA} by simply replacing $A$ with $\A\Psi:=A\Psi\oplus\Psi$, so that $\ker(\A)=\{0\}$. Even if, with such a replacement, the hypothesis (h2) in \cite{JFA} is violated (see Remark \ref{2.5} below), and $\widehat H$ is not representable as a self-adjoint extension of a restriction of $H$ to a dense set, the proof of \cite[Theorem 2.1]{JFA} can be modified to adapt to the present situation. In this way, we easily obtain a resolvent formula for a self-adjoint operator which turns out to have the same domain and action as the one provided in \cite{LS}, i.e., as in \eqref{DL} and \eqref{AL} (see Theorems  \ref{TK} and \ref{real}). Furthermore, we improve the results provided in \cite{MPAG} in that we no longer need to assume that  the range of $A$ is dense (this originates from an issue raised by Sascha Lill). The range is dense for the Nelson model and in the general abstract case this is equivalent to a trivial intersection between the domains of the free Hamiltonian and the interacting one (see Remark \ref{r-dom}).\par
The found resolvent formula, together with a similar one holding for the regularized operators $H_{\Lambda}=H+A_{\Lambda}^{*}+A_{\Lambda}$ (see Lemma \ref{krf-n}), leads  to a 
convergence result in norm resolvent sense as $\Lambda\nearrow\infty$ (see Theorem \ref{teo-conv}), allowing an approximation scheme which fits into the framework of the regularization of quantum fields Hamiltonians by an ultraviolet cutoff.\par Our abstract setting applies to various renormalizable models in Quantum Field Theory; for such results, we refer to \cite[section 3.1]{MPAG}. The main technical point regards the existence of the limit \eqref{ARA*-TS} in Theorem \ref{teo-conv} below (this corresponds to \cite[Theorem 3.10]{MPAG}). Such a not trivial result is provided in \cite{LS} as regards the Nelson model and in \cite{L2}, \cite{L3}, \cite{Sch1}, \cite{Sch2} for other related models; for that, the Fock space structure of $\F$ finally comes into play. 
\end{section}
\subsection{Notation and definitions.} \begin{itemize}
\item $\D(L)$, $\ker(L)$, $\ran(L)$ denote the domain, kernel and range of the linear operator $L$ respectively; 
\item $\varrho(L)$ denotes the resolvent set of $L$;
\item $L|V$ denotes the restriction of $L$ to the subspace $V\subset\D(L)$;
\item $\B(X,Y)$ denotes the set of bounded linear operators on the Banach space $X$ to the Banach space $Y$, $\B(X):=\B(X,X)$; 
\item $\|\cdot\|_{X,Y}$ denotes the norm in $\B(X,Y)$;
\item A linear operator $S$ is said to be $H$-small whenever $\dom(S)\supseteq\dom(H)$ and there exist $ a\in[0,1)$ and $b\in\RE$ such that $\|S\psi\|\le a\,\|H\psi\|+b\,\|\psi\|$ for any $\psi\in\dom(H)$;
\item $S_{n}$ is said to be uniformly $H_{n}$-small whenever each $S_{n}$ is $H_{n}$-small with constants $a\in[0,1)$ and $b\in\RE$ which are $n$-independent.
\end{itemize}
\begin{section}{Building a resolvent}
\noindent Let $$H :\dom(H)\subseteq\F\to\F$$ be a bounded-from-below self-adjoint operator in the Hilbert space $\F$ equipped with the scalar product $\langle\cdot,\cdot\rangle$ and corresponding norm $\|\cdot\|$. We denote by $\Sob_{s}$, $s\in [-1,1]$, the scale of Hilbert spaces given by the completion of $\dom(H)$ endowed with the scalar product
$$
\langle \psi_{1},\psi_{2}\rangle_{s}:=\langle(H^{2}+1)^{s/2}\psi_{1},(H^{2}+1)^{s/2}\psi_{2}\rangle\,.
$$
Obviously, $\H_{0}\equiv\F$ and $\H_{1}\equiv\dom(H)$. One has
\be\label{incl}
\H_{s}\hookrightarrow\F\hookrightarrow\H_{-s}\,,\qquad 0<s\le 1\,, 
\ee
with dense inclusions; moreover, the interpolation theorems hold for the scale $\H_{s}$ (see \cite[Section 9]{KP}). By the dense inclusions \eqref{incl}, the scalar product on $H$ induces the dual pairings (conjugate-linear with respect to the first variable)
$\langle\cdot,\cdot\rangle_{\pm s,\mp s}$ between the dual couples
$(\H_{\mp s},\H_{\pm s})$; these induce the dual pairings $\langle\!\langle\cdot,\cdot\rangle\!\rangle_{\pm s,\mp s}$ between the dual couples
$(\F\oplus\H_{\mp s},\F\oplus\H_{\pm s})$. In the following, adjoints are taken with respect to such dualities.\par  
Let
\be\label{ann}  
 A:{\Sob_1} \to \F\,,
\ee 
be a bounded linear map; for any $z\in  \varrho(H)$ we define the bounded operator 
\be\label{Gz}
G_{z} : \F\to\F \,,\qquad G_{z}:=(  A R_{{\bar z} })^{*}\,,
\ee
where
$$
R_{z}:\F\to{\Sob_1}\,,\qquad R_{z}:=(-H +z )^{-1}\,.
$$
Notice that $R_{z}$  extends to a bounded map $R_{z}\!:\!\H_{-1}\to\F$ with bounded inverse
$$(-\cH+z): \F\to\H_{-1}\,,
$$ where $\cH$ denotes the closure of the densely-defined, bounded operator 
$H\!:\!\H_{1}\subseteq\F\to\H_{-1}$. By introducing the adjoint 
\be\label{A*}
A^{*}:\F\to \H_{-1}\,,
\ee
one has
\be\label{RA*}
G_{z}=R_{z}A^{*}
\ee
and
\be\label{HG}
(-\cH+z)G_{z}=A^{*}\,.
\ee
By the first resolvent identity, one gets
\be\label{RG}
(z-w)R_{w}G_{z}=G_{w}-G_{z}=(z-w)R_{z}G_{w}\,,
\ee
so that 
\be\label{wz}
\ran(G_{w}-G_{z})\subseteq{\Sob_1}\,.
\ee
In the following we take $A\in\B(\H_{1},\F)$ such that 
\be\label{H3}
\ran(G_{z})\cap\H_{1}=\{0\}
\ee
for any $z\in\varrho(H)$. 
\begin{remark}\label{REM}  The inclusion \eqref{wz} implies that \eqref{H3} holds for any $z\in\varrho(H)$ whenever it holds for a single $z_{\circ}\in\varrho(H)$. By \eqref{RA*}, since $R_{z}:\H_{-1}\to\F$ maps $\F$ onto $\H_{1}$, 
$$
\ran(G_{z})\cap\H_{1}=\{0\}\quad\Leftrightarrow\quad\ran(A^{*})\cap\F=\{0\}\,.$$ 
By \cite[Lemma 2.5]{MPAG}, 
$$
\text{$\ker(A)$ is dense in $\F$}\quad\Rightarrow\quad\ran(G_{z})\cap\H_{1}=\{0\}
$$
and, by \cite[Lemma 2.11]{P03}, $\Rightarrow$ can be replaced by $\Leftrightarrow$ whenever $\ran(A)$ is dense in $\F$ then .  
\end{remark}
\begin{remark}\label{closable}
By \cite[Remark 2.11]{MPAG}, \eqref{H3} fails whenever $A\not=0$ is closable as an operator in $\F$ with $\dom(A)=\H_{1}$.   
\end{remark}
Given $A$ as above, we introduce the bounded operators
$$
\A:\H_{1}\to\F\oplus\H_{1}\,,\qquad \A\psi:=A\psi\oplus\psi
$$
and, for any $z\in\varrho(H)$,
$$
\G_{z}:\F\oplus\H_{-1}\to\F\,,\qquad\G_{z}:=(\A R_{\bar z})^{*}\,.
$$
By 
\be\label{GGz}
\G_{z}(\psi\oplus\varphi)=G_{z}\psi+R_{z}\varphi
\ee 
and by \eqref{RG}, one has
\be\label{RGG}
(z-w)R_{w}\G_{z}=\G_{w}-\G_{z}=(z-w)R_{z}\G_{w}\,.
\ee
Now, let us pick $\lambda_{\circ}\in\RE$ such that $$\lambda_{\circ}<\lambda_{\rm inf}:=\inf\sigma (H)\,,$$ and set  
\be\label{WF}
\Ml_{z}:=  \A(\G-\G_{z}): \F\oplus\H_{-1}\to \F\oplus\H_{1}\,,\qquad 
\G:=\G_{\lambda_{\circ}}\,.
\ee 
By  \eqref{RGG}, the linear operator $\Ml_{z}$ is well defined, bounded and 
\be\label{Mz}
\Ml_{z}=(z-\lambda_{\circ})\G^{*}\G_{z}=(z-\lambda_{\circ})\G_{{\bar z} }^{*}\G\,.
\ee
By \eqref{RGG} and \eqref{Mz}, one gets the relations 
\be\label{QFT}
\Ml_{z}-\Ml_{w}=(z-w)\G^{*}_{\bar w}\G_{z}\,,\qquad \Ml_{z}^{*}=\Ml_{{\bar z} }\,.
\ee
Exploiting the definitions of $\A$ and $\G_{z}$, one can re-write $\Ml_{z}$ as a sum of block operator matrices which isolates the $z$-independent component:
\begin{align}\label{ms}
\Ml_{z}
=\begin{bmatrix} 0&G^{*}\\G&R  
\end{bmatrix}+\begin{bmatrix} A(G-G_{z})&-G^{*}_{\bar z}\\
-G_{z}&-R_{z}
\end{bmatrix}\,,\qquad R:=R_{\lambda_{\circ}}\,,\quad G:=G_{\lambda_{\circ}}\,.
\end{align}
Notice that, unless $A\in\B(\H_{1/2},\F)$, it is not possible (by \eqref{H3}) to untangle the bounded (by \eqref{wz}) component $A(G-G_{z})$: the introduction of the ($\lambda_{\circ}$-dependent) operator $G$ has the main purpose of ''renormalize'' the otherwise ill-defined tern $AG_{z}$.
\begin{remark} The semi-boundedness hypothesis of $H$ is not essential; is used here for reasons of simplicity. At the expense of more complex formulae, similar results could also be obtained in the case where $\sigma(H)=\RE$. In that case, $\lambda_{\circ}$ is replaced by the imaginary unit $i$ and the role of $G$ is played by $\frac12(G_{i}+G_{-i})$.  
\end{remark}
\begin{lemma}\label{bv} If $A\in\B(\H_{s},\F)$ for some $s\in(0,1)$, then both $(1-G_{\lambda})$ and $(1-G_{\lambda}^{*})$ have bounded inverses in $\F$ whenever $\lambda< \lambda_{\rm inf}-\max\big\{\sqrt{ \lambda^{2}_{\rm inf}+1},\|A\|^{s_{*}}_{\H_{s},\F}\big\}$, $s_{*}:=\frac1{1-s}$. \end{lemma}
\begin{proof}
Let $\lambda<\lambda_{\rm inf}$. Since 
$$\|R_{\lambda}\|_{\F,\F}\leq ( \lambda_{\rm inf}-\lambda)^{-1}\,,\qquad\|R_{\lambda}\|_{\F,\H_{1}}\leq1\,,\qquad \lambda\le \lambda_{\rm inf}-\sqrt{ \lambda^{2}_{\rm inf}+1}\,,
$$ 
one gets, by interpolation, 
\be\label{interp}
\|R_{\lambda}\|_{\F,\H_{s}}\leq( \lambda_{\rm inf}-\lambda)^{s-1}\,,\qquad 0\le s\le 1\,.
\ee
Therefore, if $A\in\B(\H_{s},\F)$, $s<1$, then one has  $\|G_{\lambda}\|_{\F,\F}=\|G_{\lambda}^{*}\|_{\F,\F}<1$ whenever $\lambda< \lambda_{\rm inf}-\max\big\{\sqrt{ \lambda^{2}_{\rm inf}+1},\|A\|^{s_{*}}_{\H_{s},\F}\big\}$. Hence, for such a $\lambda$, both $(1-G_{\lambda})$ and $(1-G_{\lambda}^{*})$ have bounded inverses.
\end{proof}
\begin{remark}\label{ig} If, in the proof of Lemma \ref{bv}, one replaces \eqref{interp} with 
$$
\|R_{i\gamma}\|_{\F,\H_{s}}\leq |\gamma|^{\,s-1}\,,\qquad \gamma\in\RE\,,\ |\gamma|\ge 1\,,\quad   0\le s\le 1\,,
$$
then one gets that  both $(1-G_{i\gamma})$ and $(1-G_{i\gamma}^{*})$ have bounded inverses in $\F$ whenever $|\gamma|>\max\big\{1,\|A\|^{s_{*}}_{\H_{s},\F}\big\}$.
\end{remark}
Since the first operator matrix in \eqref{ms} is symmetric and $z$-independent, the second one satisfies the same relations \eqref{QFT} as $\Ml_{z}$; moreover, the first operator matrix in \eqref{ms} remains symmetric if its zero entry is replaced by an arbitrary symmetric operator. Furthermore, Lemma \ref{bv}  suggests to replace $G$ and $G^{*}$ with $1-G$ and $1-G^{*}$. Hence, we re-write $\Ml_{z}$ as the sum
$$\Ml_{z}
=\begin{bmatrix} -T&G^{*}-1\\G-1&R  
\end{bmatrix}+\begin{bmatrix} T+A(G-G_{z})&1-G^{*}_{\bar z}\\
1-G_{z}&-R_{z}
\end{bmatrix}.
$$ 
For later convenience,  the symmetric $T$ is chosen of the form
\be\label{TS}
T_{S}:=(1-G^{*})S(1-G)\,,
\ee
where $S$ is symmetric as well; this allows the factorization
$$
\begin{bmatrix} -T_{S}&G^{*}-1\\G-1&R  
\end{bmatrix}=
\begin{bmatrix}1-G^{*}&0\\0&1\end{bmatrix}\begin{bmatrix}{-S}&-\uno\\-\uno&R  \end{bmatrix}\begin{bmatrix}1-G&0\\0&1
\end{bmatrix}\,.
$$
Therefore, following the scheme provided in \cite{JFA}, given any symmetric operator 
$$
S:\D(S)\subseteq\F\to\F\,,\qquad\D(S)\supseteq \H_{1}\,,
$$  
we define the operator $$
\Thl_{S}:{\S}\oplus\F\subseteq\F\oplus\H_{-1}\to\F\oplus\H_{1}
$$
by \be\label{ts}
\S:=\{\psi\in\F: (1-G)\psi\in\H_{1}\},
\qquad
\Thl_{S}:=\begin{bmatrix} T_{S}&1-G^{*}\\1-G&-R  
\end{bmatrix}\,,
\ee
and consider the operator 
$$
\Thl_{S}+\Ml_{z}
:{\S}\oplus\F\subseteq\F\oplus\H_{-1}\to\F\oplus\H_{1}
$$
\be\label{thl}
\Thl_{S}+\Ml_{z}=\begin{bmatrix} T_{S}+A(G-G_{z})&1-G^{*}_{\bar z}\\
1-G_{z}&-R_{z}
\end{bmatrix}\,.
\ee
Since $\Thl_{S}$ is symmetric, i.e., for any $\psi_{1}\oplus\phi_{1}$ and $\psi_{2}\oplus\phi_{2}$ in ${\S}\oplus\F$,
\be\label{symm}
\langle\!\langle (\psi_{1}\oplus\phi_{1}),\Thl_{S}(\psi_{2}\oplus\phi_{2})\rangle\!\rangle_{-1,1}
=\langle\!\langle \Thl_{S}(\psi_{1}\oplus\phi_{1}),(\psi_{2}\oplus\phi_{2})\rangle\!\rangle_{1,-1}\,,
\ee
the operator in \eqref{thl} satisfies, by \eqref{QFT},  the same kind of relations as $\Ml_{z}$:
\be\label{QF1}
(\Thl_{S}+\Ml_{z})-(\Thl_{S}+\Ml_{w})
=(z-w)\G^{*}_{\bar w}\G_{z}\ee
and
\be\label{QF2}
\langle\!\langle (\psi_{1}\oplus\phi_{1}),(\Thl_{S}+\Ml_{z})(\psi_{2}\oplus\phi_{2})\rangle\!\rangle_{-1,1}
=\langle\!\langle (\Thl_{S}+\Ml_{\bar z})(\psi_{1}\oplus\phi_{1}),(\psi_{2}\oplus\phi_{2})\rangle\!\rangle_{1,-1}\,.
\ee
\begin{remark} Obviously, whenever $\S$ is dense, \eqref{symm} re-writes as $\Thl_{S}\subseteq\Thl_{S}^{*}$ and \eqref{QF2} re-writes as $(\Thl_{S}+\Ml_{\bar z})\subseteq (\Thl_{S}+\Ml_{z})^{*}$.
\end{remark}
Finally, we define 
\be\label{ZS}
Z_{S}:=\big\{z\in\varrho(H):\big(\Thl_{S}+\Ml_{z}\big)^{-1}\in\B(\F\oplus\H_{1},\F\oplus\H_{-1})\big\}
\,.
\ee
Notice that, by \eqref{WF}, 
\be\label{iff}
\lambda_{\circ}\in Z_{S}\quad\Leftrightarrow\quad \Thl_{S}^{-1}\in\B(\F\oplus\H_{1},\F\oplus\H_{-1})\,.
\ee
\begin{lemma}\label{supp} Suppose $\lambda_{\circ}\in Z_{S}$. Then $Z_{S}$ is an open subset of $\CO$ symmetric with respect to $\RE$, i.e., $z\in Z_{S}$ if and only if $\bar z\in Z_{S}$; moreover, ${\S}$ is dense  and $\Thl_{S}$ is self-adjoint.
\end{lemma}
\begin{proof} The set $Z_{S}$ is open in  $\CO$ by the continuity of the map $z\mapsto \Ml_{z}$. By our hypothesis and by \eqref{symm}, $\Thl_{S}^{-1}$ is bounded and symmetric, hence self-adjoint. The same hypothesis also gives 
$$\overline{{\S}\oplus\F}=\overline{\S}\oplus\H_{-1}=\overline{\dom(\Thl_{S})}=\overline{\ran(\Thl_{S}^{-1})}=
\ker(\Thl_{S}^{-1})^{\perp}=\{0\}^{\perp}=\F\oplus\H_{-1}\,,$$ and so ${\S}$ is dense in $\F$; then, the densely-defined $\Thl_{S}$ is self-adjoint as the inverse of a self-adjoint operator. \par
By $\big(\Thl_{S}+\Ml_{z}\big)^{*}=\Thl_{S}+\Ml_{\bar z}$ and and by \cite[Thm. 5.30, Chap. III]{K}, one gets that $z\in Z_{S}$ is equivalent to $\bar z\in Z_{S}$.
\end{proof}
\begin{theorem}\label{TK} Let $H:\H_{1}\subseteq\F\to\F$  bounded-from-below  and let $A\in\B(\H_{1},\F)$ satisfy \eqref{H3}. If $\lambda_{\circ}\in Z_{S}$, then
\be\label{krein}
R_{z}+\G_{z}\big(\Thl_{S}+\Ml_{z}\big)^{-1}\G_{\bar z}^{*}\,,\qquad z\in Z_{S}\,,
\ee
is the resolvent of a self-adjoint operator. 
\end{theorem}
\begin{proof} Let us set
$$
{\mathbb\Lambda}_{z}:=\big(\Thl_{S}+\Ml_{z}\big)^{-1}\,,\qquad
\widehat R_{z}:=R_{z}+\G_{z}{\mathbb\Lambda}_{z}\G_{\bar z}^{*}\,,\quad z\in Z_{S}
\,.
$$
By \eqref{QF1}, \eqref{QF2} and by the same simple calculations as in \cite[page 115]{JFA} (with $\tau=\A$, $\mathscr X=\F\oplus\H_{1}$ and $\Lambda(z)={\mathbb\Lambda}_{z}$),
$\widehat R_{z}$ is a pseudo-resolvent, i.e., it satisfies $\widehat R_{z}-\widehat R_{w}=(w-z)\widehat R_{z}\widehat R_{w}$. Thus, by \cite[Thm. 4.10]{St}$, \widehat R_{z}$ is the resolvent of a closed operator whenever $\ker(\widehat R_{z_{\circ}})=\{0\} $ for some $z_{\circ}\in Z_{S}$. Let us take $z_{\circ}=\lambda_{\circ}$,  $\psi_{\circ}\in\ker(\widehat R_{\lambda_{\circ}}  )$ and define 
$\phi_{\circ}\oplus\varphi_{\circ}:={\mathbb\Lambda}_{\lambda_{\circ}}\G^{*}\psi_{\circ}=\Thl_{S}^{-1}\G^{*}\psi_{\circ}$. Then, by \eqref{GGz} and by $R\psi_{\circ}+\G(\phi_{\circ}\oplus\varphi_{\circ})=0$, one gets 
$$
R  (\psi_{\circ}+\varphi_{\circ})+G\phi_{\circ}=0\,.
$$
By \eqref{H3}, this entails $G\phi_{\circ}=0$ and $\psi_{\circ}+\varphi_{\circ}=0$. Thus, by  \eqref{ts},
\begin{align*}
G^{*}\psi_{\circ}\oplus R  \psi_{\circ}=&\G^{*}\psi_{\circ}=\Thl_{S}(\phi_{\circ}\oplus\varphi_{\circ})=
\Thl_{S}(\phi_{\circ}\oplus(-\psi_{\circ}))\\
=&(T_{S}\phi_{\circ}+G^{*}\psi_{\circ}-\psi_{\circ})\oplus (\phi_{\circ}+R  \psi_{\circ})\,.
\end{align*}
This gives $\phi_{\circ}=0$ and $\psi_{\circ}=0$; hence $\ker(\widehat R_{\lambda_{\circ}})=\{0\}$. Finally, since, by \eqref{QF2}, there holds ${\mathbb\Lambda}_{z}^{*}={\mathbb\Lambda}_{\bar z}$, one gets $\widehat R_{z}^{*}=\widehat R_{\bar z}$ and so the closed operator corresponding to $\widehat R_{z}$ is densely defined and self-adjoint. \end{proof}
\begin{remark}\label{2.5} While Theorem \ref{TK} would seem to be a corollary of \cite[Thm. 2.1]{JFA}, let us notice that the hypothesis (h2) there, which, by \cite[Rem. 2.8]{JFA}, is equivalent to 
$\ran(\G_{z})\cap\H_{1}=\{0\}$, is here violated. Indeed, by \eqref{GGz}, one has 
$\ran(\G_{z})\supseteq\H_{1}$. If hypothesis (h2) in \cite[Thm. 2.1]{JFA} were true, then, by the aforementioned equivalence, the relation $R_{z}\psi=-\G_{z}{\mathbb\Lambda}_{z}\G_{\bar z}^{*}\psi$, holding for any $\psi\in\ker(\widehat R_{z})$, would immediately give $\psi=0$,
\end{remark}
The proof of Theorem \ref{TK} strongly relies on the assumption $\lambda_{\circ}\in Z_{S}$, which seems not so easy to check. However, one has a simple criterion in terms of the $H$-smallness of $S$. 
\begin{lemma}\label{KR} Suppose that $A\in\B(\H_{s},\F)$ for some $s\in(0,1)$ and that
the symmetric operator $S$ is $H$-small, $$\|S\psi\|\le a\,\|H\psi\|+b\,\|\psi\|\,,\qquad a\in[0,1)\,,\quad b\in\RE\,.
$$ Then  
$$
\lambda_{\circ}< \lambda_{\rm inf}-\max\left\{\sqrt{ \lambda^{2}_{\rm inf}+1}\,,\|A\|_{\H_{s}\,,\F}^{s_{*}}\,,\frac{b}{1-a}\right\}\quad\Rightarrow\quad \lambda_{\circ}\in Z_{S}\,.
$$ 
\end{lemma}
\begin{proof} By Lemma \ref{bv} , $1-G$ and $1-G^{*}$ have bounded inverses whenever $$\lambda_{\circ}< \lambda_{\rm inf}-\max\big\{\sqrt{ \lambda^{2}_{\rm inf}+1}\,,\|A\|_{\H_{s}\,,\F}^{s_{*}}\big\}\,.$$ 
Hence,  by the factorization
\be\label{GSG}
\Thl_{S}=
\begin{bmatrix}1-G^{*}&0\\0&1\end{bmatrix}{\mathbb S}
\begin{bmatrix}1-G&0\\0&1
\end{bmatrix}\,,\qquad {\mathbb S}:=\begin{bmatrix}{S}&\uno\\\uno&-R  \end{bmatrix}\,,
\ee
$\Thl_{S}$ has a bounded inverse, equivalently $\lambda_{\circ}\in Z_{S}$, whenever ${\mathbb S}^{-1}\in\B(\F\oplus\H_{1},\F\oplus\H_{-1})$. Since $S$ is $H$-small,
\be\label{res-est}
\|SR\|_{\F,\F}\le a+\frac{b}{ \lambda_{\rm inf}-\lambda_{\circ}}\,,
\qquad \lambda_{\circ}\le\min\{0,2\lambda_{\rm inf}\}\,.
\ee
The second Schur complement of ${\mathbb S}$ is given by $S+R  ^{-1}=-(H-S)+\lambda_{\circ}$; by \eqref{res-est}, it has the bounded inverse $$R_{S}:=(-(H-S)+\lambda_{\circ})^{-1}=R  (1+SR  )^{-1}=(1+R  S)^{-1}R  \in\B(\F,\H_{1})\,,\qquad \lambda_{\circ}< \lambda_{\rm inf}-\frac{b}{1-a}\,.
$$  
Thus,
\begin{align}\label{S-1}
{\mathbb S}^{-1}=&
\begin{bmatrix}
R_{S}&R_{S}(-H+\lambda_{\circ})\\
(-H+\lambda_{\circ})R_{S}&
(-H+\lambda_{\circ})(R_{S}-R  )(-H+\lambda_{\circ})\end{bmatrix}
\in \B(\F\oplus\H_{1},\H_{1}\oplus\F)\,.
\end{align}
\end{proof}
\begin{remark}\label{S=0} In the case $S=0$, one has 
$$
{\mathbb S}^{-1}=
\begin{bmatrix}
R&1\\1&0
\end{bmatrix}\in \B(\F\oplus\F)
$$
and so $\Thl_{0}^{-1}\in\B(\F\oplus\F)$.
\end{remark}
\begin{remark}\label{sub} Suppose that the $H$-small symmetric operator  $S$ in Lemma \ref{KR} depends on $\lambda_{\circ}$ in such a way that $a<1$ uniformly w.r.t. $\lambda_{\circ}$ and $b$ grows at most sub-linearly as $|\lambda_{\circ}|\nearrow\infty$. Then, by the same proof as above (in particular, see \eqref{res-est}), still   
$\lambda_{\circ}\in Z_{S}$, whenever $\lambda_{\circ}$ is sufficiently below  $ \lambda_{\rm inf}$. 
\end{remark}
\end{section}
\begin{section}{Self-adjoint realizations}
In this section we give an explicit representation of the self-adjoint operator provided in Theorem \ref{TK}. Here, we do not anymore need the hypothesis $\lambda_{\circ}\in Z_{S}$ and such a  representation holds for an arbitrary $\lambda_{\circ}< \lambda_{\rm inf}=\inf\sigma(H)$; in the sense specified by the Remark \ref{indep} below, the operator is independent of the choice of $\lambda_{\circ}$.
\begin{theorem}\label{real} Let  $H:\H_{1}\subseteq\F\to\F$ be self-adjoint and  bounded-from-below and let the symmetric operator $S$ be $H$-small; let $A:\H_{1}\to\F$ belong to $\B(\H_{s},\F)$ for some $s\in(0,1)$ and suppose that \eqref{H3} holds true. Then 
$$
H_{S}:\S\subseteq\F\to\F\,,\qquad H_{S}:=\cH+A^{*}+A_{S}\,,
$$
$$
\S =\{\psi\in\F:\psi_{\circ}:=\psi-G\psi\in\H_{1}\}\,,
$$
$$
A_{S}:\S \subseteq\F\to\F\,,\qquad A_{S}\psi
:=A\psi_{\circ}-T_{S}\psi\,,\quad T_{S}:=(1-G^{*})S(1-G)\,,
$$
is a bounded-from-below self-adjoint operator with resolvent
\be\label{krf}
(-H_{S}+z)^{-1}=R_{z}+\G_{z}\big(\Thl_{S}+\Ml_{z}\big)^{-1}\G_{\bar z}^{*}\,,\qquad z\in \varrho(H)\cap\varrho(H_{S})\,.
\ee
\end{theorem}
\begin{proof} At first, let us suppose that $\lambda_{\circ}$ is sufficiently below  $ \lambda_{\rm inf}$  so that, by Lemma \ref{KR} and \eqref{iff}, $\Thl_{S}$ has a bounded inverse. Then, by \eqref{krein} with $z=\lambda_{\circ}$, the self-adjoint operator provided in Theorem \ref{TK}, here denoted by $H_{S}^{\circ}$, has domain 
\be\label{dHS}
\dom( H^{\circ}_{S})=\{\psi\in\F:\psi=\phi+\G\Th_{S}^{-1}\A\phi,\ \phi\in\H_{1}\}\,.
\ee
Writing $\Th_{S}^{-1}\A\phi=\varphi_{0}\oplus\varphi_{1}\in \S\oplus\F$, so that $\A\phi=\Th_{S}(\varphi_{0}\oplus\varphi_{1})$, by the definitions of $\Th_{S}$ and $\A$, one gets
\be\label{T-inv}
(1-G^{*})S(1-G)\varphi_{0}+(1-G^{*})\varphi_{1}=A\phi\,,\qquad (1-G)\varphi_{0}-R\varphi_{1}=\phi\,.
\ee
Therefore, if $\psi\in\dom(H_{S})$, then 
$$
\psi=\phi+G\varphi_{0}+R\varphi_{1}=\varphi_{0}=\phi+R\varphi_{1}+G\psi
$$
and so 
\be\label{ppvp}
\psi-G\psi=\phi+R\varphi_{1}\in\H_{1}\,.
\ee
This gives $\dom( H^{\circ}_{S})\subseteq \S $. Conversely, let $\psi\in\F$ be such that $\psi_{\circ}:=\psi-G\psi\in\H_{1}$.  Thus, for any given $\varphi_{1}\in\F$, one has 
$$
\psi=\phi+\G(\psi\oplus\varphi_{1})\,,\qquad\phi:=\psi_{\circ}-R\varphi_{1}\in\H_{1}\,. 
$$
Since $\psi\oplus\varphi_{1}\in\S\oplus\F=\dom(\Th_{S})$, there exists an unique  $\phi_{0}\oplus\phi_{1}\in\F\oplus\H_{1}$ such that $\Th_{S}(\psi\oplus\varphi_{1})=\phi_{0}\oplus\phi_{1}$. Thus, we need to show that there exists $\varphi_{1}\in\F$ such that $\Th_{S}(\psi\oplus\varphi_{1})=\A\phi$. Equivalently, see \eqref{T-inv}, $\varphi_{1}$ has to solve both the equations
$$
(1-G^{*})S(1-G)\psi+(1-G^{*})\varphi_{1}=A\phi\,,\qquad (1-G)\psi-R\varphi_{1}=\phi\,.
$$
However, by the definition of $\psi_{\circ}$ and $\phi$, the second equation reduces to an identity and so it suffices to solve the first one, which rewrites as
$$
\varphi_{1}=(1-G^{*})^{-1}(A\psi_{\circ}-G^{*}\varphi_{1})-S\psi_{\circ}\,.
$$
Such an equation has solution 
\be\label{sol}
\varphi_{1}=\big(A-(1-G^{*})S\big)\psi_{\circ}\in\F\,.
\ee
This gives $\S \subseteq \dom( H^{\circ}_{S})$ and so $\dom( H^{\circ}_{S})=\S$.\par By the domain representation \eqref{dHS} and by \eqref{krein}, one gets
$$
(-H^{\circ}_{S}+\lambda_{\circ})\psi=(-H+\lambda_{\circ})\phi\,.
$$
Therefore, by \eqref{ppvp} and \eqref{sol},
$$
(-H^{\circ}_{S}+\lambda_{\circ})\psi=(-\cH+\lambda_{\circ})\psi-A^{*}\psi-\varphi_{1}=
(-\cH+\lambda_{\circ})\psi-A^{*}\psi-A\psi_{\circ}+T_{S}\psi\,.
$$
Hence,
$$
H^{\circ}_{S}=\cH+A^{*}+A_{S}
$$
and so $H_{S}=H^{\circ}_{S}$ is self-adjoint.\par
Let us now suppose that $\lambda_{\circ}$ is any point below $ \lambda_{\rm inf}$ and let us define the symmetric operator 
\be\label{wts}
\widetilde S:=(1-G_{\lambda}^{*})^{-1}\big(A( G -G_{\lambda})+T_{S}\big)(1- G_{\lambda} )^{-1}\,,
\ee
where $\lambda $ is sufficiently below  $ \lambda_{\rm inf}$  so that $(1- G_{\lambda} )$ and $(1-G_{\lambda}^{*})$ have bounded inverses (see Lemma \ref{bv} ). 
By \eqref{RG}, one gets
\begin{align*}
\widetilde S=&(1-G_{\lambda}^{*})^{-1}A( G -G_{\lambda})(1- G_{\lambda} )^{-1}\\
&+(1-G_{\lambda}^{*})^{-1}(1-G_{\lambda}^{*}+(G_{\lambda}^{*}-G^{*}))S(1- G_{\lambda} +( G_{\lambda} -G))(1- G_{\lambda} )^{-1}\\
=&(\lambda-\lambda_{\circ} )(1-G_{\lambda}^{*})^{-1}G_{\lambda}^{*}G(1- G_{\lambda} )^{-1}\\
&+\big(1-(\lambda-\lambda_{\circ} )(1-G_{\lambda}^{*})^{-1}G_{\lambda}^{*}R\big)S
\big(1-(\lambda-\lambda_{\circ} )R G_{\lambda} (1- G_{\lambda} )^{-1}\big)\\
=&S+(\lambda-\lambda_{\circ} )(1-G_{\lambda}^{*})^{-1}G_{\lambda}^{*}G(1- G_{\lambda} )^{-1}\\
&-(\lambda-\lambda_{\circ} )(1-G_{\lambda}^{*})^{-1}G_{\lambda}^{*}RS-
(\lambda-\lambda_{\circ} )SR G_{\lambda} (1- G_{\lambda} )^{-1}\\
&-(\lambda-\lambda_{\circ} )^{2}(1-G_{\lambda}^{*})^{-1}G_{\lambda}^{*}RSR G_{\lambda} (1- G_{\lambda} )^{-1}\,.
\end{align*}
Since $S$ is symmetric and $H$-bounded, one has $SR\in\B(\F)$ and $RS\subseteq RS^{*}\subseteq (SR)^{*}\in\B(\F)$; therefore,  
\be\label{BF}
\widetilde S-S\in\B(\F)
\ee
and so 
$$
\|\widetilde S\psi\|\le a\,\|H\psi\|+\widetilde b\,\|\psi\|\,, \qquad \widetilde b:=b
+\|\widetilde S-S\|_{\F,\F}\,.
$$
Since, by Lemma \ref{bv} , $|\lambda |\| G_{\lambda} \|_{\F,\F}=|\lambda |\|G_{\lambda}^{*}\|_{\F,\F}$ grows as $|\lambda |^{s}$, $s<1$, one has that $\widetilde b$ grows sub-linearly as $|\lambda |\nearrow\infty$. Hence, by Remark \ref{sub}, $\lambda \in Z_{\widetilde S}$ whenever $\lambda $ is sufficiently below  $ \lambda_{\rm inf}$. Therefore, Theorem \ref{TK} provides a self-adjoint operator  $\widetilde H^{\circ}_{\widetilde S}$ which, by what we just proved above, coincides with $$
\widetilde H_{\widetilde S}:\dom(\widetilde H_{\widetilde S})\subseteq\F\to\F\,,\qquad \widetilde H_{\widetilde S}=\cH+A^{*}+A_{\widetilde S}\,,
$$
$$
\dom(\widetilde H_{\widetilde S})=\{\psi\in\F:\widetilde \psi_{\circ}:=\psi- G_{\lambda} \psi\in\H_{1}\},
$$
where
$$
A_{\widetilde S}\psi:=A\widetilde \psi_{\circ}-\widetilde T_{\widetilde S}\,\psi\,,\qquad 
\widetilde T_{\widetilde S}:=(1-G_{\lambda}^{*})\widetilde S(1- G_{\lambda} )
\,.
$$ 
By \eqref{wz}, one has $\ran(G- G_{\lambda} )\subseteq \H_{1}$ and so $\dom(\widetilde H_{\widetilde S})=\S$.  By the definition \eqref{wts}, one gets 
\begin{align}\label{StS}
&\widetilde H_{\widetilde S}\,\psi=
(\cH+A^{*})\psi+A\widetilde\psi_{\circ}-\widetilde T_{\widetilde S}\,\psi\nonumber\\
=&(\cH+A^{*})\psi+A\widetilde\psi_{\circ}+A( G_{\lambda} -G)\psi-T_{S}\psi\nonumber\\
=&(\cH+A^{*})\psi+A\psi_{\circ}+(1- G^{*})S(1- G)\psi\\
=&(\cH+A^{*}+A_{S})\psi\nonumber\\
=&H_{S}\psi\nonumber\,.
\end{align}
Now, let $\widetilde \Thl_{\widetilde S}+\widetilde \Ml_{z}$ be the operator in \eqref{thl} corresponding to $\lambda $ and $\widetilde S$. Then
\begin{align}\label{thl1}
\widetilde \Thl_{\widetilde S}+\widetilde \Ml_{z}=&\begin{bmatrix} \widetilde T_{\widetilde S}+A( G_{\lambda} -G_{z})&1- G^{*}_{\bar z}\\
1-G_{z}&-R_{z}
\end{bmatrix}
=\begin{bmatrix}  T_{S}+A(G-G_{z})&1- G^{*}_{\bar z}\\
1-G_{z}&-R_{z}
\end{bmatrix}
=\Thl_{S}+\Ml_{z}\,.
\end{align}
Therefore 
$$\lambda \in Z_{S}:=
\{z\in\varrho(H): (\Thl_{S}+\Ml_{z})^{-1}\in\B(\F\oplus\H_{1},\F\oplus\H_{-1})\}
$$ 
and, for any $z\in Z_{S}$,
\be\label{lr}
(-H_{S}+z)^{-1}=(-\widetilde H_{\widetilde S}+z)^{-1}=R_{z}+\G_{z}\big(\widetilde \Thl_{\widetilde S}+\widetilde \Ml_{z}\big)^{-1}\G_{\bar z}^{*}=R_{z}+\G_{z}\big(\Thl_{S}+\Ml_{z}\big)^{-1}\G_{\bar z}^{*}\,.
\ee
By \cite[Theorem 2.19 and Remark 2.20]{CFP}, one gets $Z_{S}=\varrho(H)\cap\varrho(H_{S})$.
Further, by the choice of $\lambda $ and by Lemma \ref{KR}, one has $\widetilde \Thl_{\widetilde S}^{-1}\in\B(\F\oplus\H_{1},\F\oplus\H_{-1})$ and so, by \eqref{lr}, 
\be\label{wt-res}
(-H_{S}+\lambda )^{-1}=(-\widetilde H_{\widetilde S}+\lambda )^{-1}=R_{\lambda}+\G_{\lambda}\widetilde \Thl_{\widetilde S}^{-1}\G_{\lambda}^{*}\,.
\ee
This gives $\lambda \in\varrho(H_{S})$; since $\lambda $ has to be sufficiently below  $ \lambda_{\rm inf}$  but is otherwise arbitrary,  $H_{S}$ is  bounded-from-below . 
\end{proof}
\begin{remark} Notice that, by \eqref{HG}, one gets $(\cH-\lambda_{\circ})(1-G)=\cH+A^{*}-
\lambda_{\circ}$, which implies that the operator $(\cH+A^{*})|\S$ is $\F$-valued. By $T_{S}=(1-G^{*})S(1-G)$, one has $A_{S}\psi=(A-(1-G^{*})S)\psi_{\circ}$ and so $A_{S}|\S$ is $\F$-valued as well. Therefore, $H_{S}$ is $\F$-valued.
\end{remark}
\begin{remark}\label{r-dom} By \eqref{wz}, the definition of $\S$ is $\lambda_{\circ}$-independent. Since $R\in\B(\H_{s-1},\H_{s})$, one has $AR\in \B(\H_{s-1},\F)$ whenever $A\in\B(\H_{s},\F)$. Then, by duality, $(AR)^{*}=G\in\B(\F,\H_{1-s})$ and so
$$
\S\subseteq\H_{1-s}\,.
$$   
Moreover, by \eqref{H3},
$$
\S\cap\H_{1}=\ker(G)=\ran(AR)^{\perp}\,.
$$
Hence, since $R:\F\to\H_{1}$ is a continuous bijection,
$$
\S\cap\H_{1}=\{0\}\quad\Leftrightarrow\quad \text{$\ran(A)$ is dense.}
$$ 
\end{remark}
\begin{remark}\label{indep} Given the self-adjoint $H_{S}$ with domain and action represented as in Theorem \ref{real} using some $\lambda_{\circ}<\lambda_{\rm inf}$,  one has, by \eqref{StS},  
$H_{S}=\tilde H_{\tilde S}$, where $\tilde H_{\tilde S}$ is the representation which uses a different $\tilde\lambda_{\circ}<\lambda_{\rm inf}$ and the symmetric $\tilde S$ defined as in \eqref{wts} with $\lambda=\tilde\lambda_{\circ}$. By $\tilde S-S\in\B(\F)$ (see \eqref{BF})), $\tilde S$ is $H$-small as well. Thus, we can view the two couples $(\lambda_{\circ}, S)$ and $(\tilde\lambda_{\circ},\tilde S)$ as members of two equivalent sets of coordinates for the manifold which parametrizes the family of our self-adjoint realizations of the formal sum $H+A^{*}+A$. 
\end{remark}
\begin{remark} The operator $\Thl_{S}$ appearing in the resolvent formula \eqref{krf} is self-adjoint. Indeed, by \eqref{thl1}, one gets $\Thl_{S}=\widetilde \Thl_{\widetilde S}+\widetilde \M_{\lambda_{\circ}}$. Since $\widetilde \Thl_{\widetilde S}$ is self-adjoint by Lemma \ref{supp} and $\widetilde \M_{\lambda_{\circ}}$ is bounded and symmetric, $\Thl_{S}$ is self-adjoint as well.
\end{remark}
\begin{remark}\label{RS=0} In \eqref{krf},  one has 
$$
(\Thl_{S}+\Ml_{z})^{-1}\in \B(\F\oplus\H_{1},\F\oplus\H_{-1})\,,\qquad z\in \varrho(H)\cap\varrho(H_{S})\,.
$$ 
By Remark \ref{S=0} and by \cite[Theorem 2.19 and Remark 2.20]{CFP},   one has, whenever $S=0$,
$$
(\Thl_{0}+\Ml_{z})^{-1}\in \B(\F\oplus\F)\,,\qquad z\in \varrho(H)\cap\varrho(H_{0})\,.
$$ 
\end{remark}
\begin{remark} If $A=0$ then $H_{S}=H-S$ and \eqref{krf} reduces to the usual resolvent formula for the perturbation by a $H$-small symmetric operator:
\begin{align*}
(-H_{S}+z)^{-1}=&R_{z}+\begin{bmatrix}0&R_{z}\end{bmatrix}\begin{bmatrix}-S&1\\1&-R_{z}\end{bmatrix}^{-1}\begin{bmatrix}0\\R_{z}\end{bmatrix}=R_{z}+(R_{z}(1+SR_{z})^{-1}-R_{z})\\
=&R_{z}(1+SR_{z})^{-1}\,.
\end{align*} 
\end{remark}
\end{section}
\begin{section}{Norm resolvent convergence}
Let 
$$
A_{n}:\dom(A_{n})\subseteq\F\to\F\,,\qquad n\ge 1\,,
$$ 
be a sequence of closable operators such that 
\be\label{1/2}
\dom(A_{n})\supseteq\H_{1/2}\,,\qquad A_{n}\in\B(\H_{1/2},\F)\,.
\ee 
Since $A_{n}$ is closable, 
$$(A_{n}^{*}+A_{n})^{*}\supset {A}^{**}_{n}+A^{*}_{n}=\overline{A}_{n}+A^{*}_{n}\supset A_{n}+A^{*}_{n}$$ and so $A_{n}^{*}+A_{n}$ is symmetric. We further suppose that 
\be\label{kr-n}
\text{$A_{n}^{*}+A_{n}$ is $H$-small}
\ee
so that, by the Rellich-Kato theorem,
\be\label{Hn}
H_{n}:=H+A_{n}^{*}+A_{n}:\H_{1}\subseteq\F\to\F
\ee
is self-adjoint and  bounded-from-below .\par
By Remark \ref{closable} and by the closability of $A_{n}$, $H_{n}$ does not match the hypotheses of Theorem \ref{real}. However, next lemma shows that its resolvent can still be represented by means of a formula having the same structure as the one in \eqref{krf}.  
\begin{lemma}\label{krf-n} Given $A_{n}$ as above, let $H_{n}:\H_{1}\subseteq\F\to\F$ be the self-adjoint operator in \eqref{Hn} and let $E_{n}:\F\to\F$ be symmetric and bounded. Then, the linear operators
$$
\A_{n}\psi:=A_{n}\psi\oplus\psi\,,\qquad \G_{n,z}:=(\A_{n}R_{z})^{*}:\H_{s}\oplus\F\to\H_{s}\,,
\qquad 0\le s\le 1\,,$$
$$
\Ml_{n,z}:=  \A_{n}(\G_{n}-\G_{n,z}):\F\oplus\F\to\F\oplus\H_{1}\,,
$$
$$
\Thl_{n}:=\begin{bmatrix} E_{n}-A_{n}RA_{n}^{*}&1-G_{n}^{*}\\1-G_{n}&-R  
\end{bmatrix}:\H_{s}\oplus\F\to\F\oplus\H_{s}\,,\qquad 0\le s\le 1\,,\qquad 
$$
$$
G_ {n,z}:=(A_{n}R_{\bar z})^{*}\,,\qquad 
G_{n}:=G_{n,\lambda_{\circ}}\,,\qquad
\G_{n}:=\G_{n,\lambda_{\circ}}\,,
$$
are bounded and
\be\label{Kn}
(-(H_{n}-E_{n})+z)^{-1}=R_{z}+\G_{n,z}\big(\Thl_{n}+\Ml_{n,z}\big)^{-1}\G_{n,\bar z}^{*}\,,\qquad z\in \varrho(H)\cap\varrho(H_{n}-E_{n})\,,
\ee
where the block operator matrix inverse exists in $\B(\F\oplus\H_{s},\H_{s}\oplus\F)$ for any $s\in[0,1]$.
\end{lemma}
\begin{proof} By \eqref{1/2} and by $R_{z}\in\B(\H_{-1/2},\H_{1/2})$, one has $A_{n}R_{z}A_{n}^{*}\in\B(\F)$. Exploiting the definitions of $\A_{n}$ and $\G_{n,z}$, one gets
$$
\M_{n,z}=\begin{bmatrix} A_{n}(G_{n}-G_{n,z})&G^{*}_{n}-G^{*}_{n,z}\\G_{n}-G_{n,z}&R-R_{z}  
\end{bmatrix}=\begin{bmatrix} A_{n}RA_{n}^{*}-A_{n}R_{z}A_{n}^{*}&(\bar z-\lambda_{\circ})G^{*}_{n,z}R\\(z-\lambda_{\circ})RG^{*}_{n,z}&R-R_{z}  
\end{bmatrix}
$$
and so $\M_{n,z}\in\B(\F\oplus\F,\F\oplus\H_{1})$.
By \eqref{kr-n}, $(A_{n}^{*}+A_{n})\in\B(\H_{1},\F)$. Thus, by $A_{n}\in\B(\H_{1/2},\F)\subset\B(\H_{1},\F)$, one gets $A_{n}^{*}\in\B(\H_{1},\F)$. Since, by duality, $A_{n}\in\B(\F,\H_{-1})$ one obtains, by interpolation, $A_{n}^{*}\in\B(\H_{s},\H_{s-1})$ and hence $G_ {n,z}\in\B(\H_{s})$, $\G_ {n,z}\in\B(\H_{s}\oplus\F,\H_{s})$. Thus,    $\Thl_{n}\in\B(\H_{s}\oplus\F,\F\oplus\H_{s})$. Using the operator block matrix representations of $\Thl_{n}$ and $\M_{n,z}$, 
one obtains 
$$
\Thl_{n}+\M_{n,z}=\begin{bmatrix} E_{n}-A_{n}RA_{n}^{*}&1-G_{n,\bar z}^{*}\\
1-G_{n,z}&-R_{z}  
\end{bmatrix}\,.
$$
Denoting by $R_{n,z}$ the second Schur complement of $\Thl_{n}+\M_{n,z}$, one has
$$
R_{n,z}=\big((1-G_{n,\bar z}^{*})R_{z}^{-1}(1-G_{n,z})+E_{n}-A_{n}RA_{n}^{*}\big)^{-1}\,.
$$
Therefore, by the identity
$$
(1-G_{n,\bar z}^{*})R_{z}^{-1}(1-G_{n,z})=
(1-G_{n,\bar z}^{*})(-H+z)(1-G_{n,z})=
-H_{n}+z+A_{n}R_{z}A_{n}^{*}\,,
$$
one gets 
\begin{align}\label{schur}
R_{n,z}=(-(H_{n}-E_{n})+z)^{-1}\,.
\end{align}
Hence, for any $z\in\varrho(H_{n}-E_{n})$, 
\begin{align}\label{TS-n}
&(\Thl_{n}+\M_{n,z})^{-1}\nonumber\\=&\begin{bmatrix} R_{n,z}&R_{n,z}(1-G_{n,\bar z}^{*})(-H+z)\\
(-H+z)(1-G_{n,z})R_{n,z}& -(-H+z)\big(1-(1-G_{n,z})R_{n,z}(1-G_{n,\bar z}^{*})(-H+z)\big) 
\end{bmatrix}\,.
\end{align}
Since $G_{n,z}\in\B(\H_{1})$ and $G^{*}_{n,z}\in\B(\F)$, one gets, by \eqref{TS-n},  $(\Thl_{n}+\M_{n,z})^{-1}\in \B(\F\oplus\H_{1},\H_{1}\oplus\F)$.  By interpolation, to show that $(\Thl_{n}+\M_{n,z})^{-1}\in \B(\F\oplus\H_{s},\H_{s}\oplus\F)$, $0\le s\le 1$, is enough to show that 
$(\Thl_{n}+\M_{n,z})^{-1}\in \B(\F\oplus\F,\F\oplus\F)$. To prove that the latter holds for any $z\in\varrho(H)\cap\varrho(H_{n}-E_{n})$ it suffices, by \cite[Theorem 2.19 and Remark 2.20]{CFP}, to show that it holds for a single point $z_{\circ}$. Here, we take $z_{\circ}=i\gamma$ and, by Remark \ref{ig}, we can choose a real $\gamma$ with $|\gamma|$ sufficiently large so that $\|G_{n,i\gamma}\|_{\F,\F}\le 
\|A_{n}\|_{\H_{1/2},\F}\,|\gamma|^{-1/2}<1/2$. Thus, both $(1-G_{n,i\gamma})$ and $(1-G^{*}_{n,-i\gamma})$ have bounded inverses in $\F$ and we can define the bounded operator in $\F$
$$
S_{n,i\gamma}:=
(1-G_{n,-i\gamma}^{*})^{-1}(E_{n}-A_{n}RA_{n}^{*})(1-G_{n,i\gamma})^{-1}\,.
$$
Furthermore, by  $\|R_{i\gamma}\|_{\H_{-1/2},\H_{1/2}}=\|R_{i\gamma}\|_{\F,\H_{1}}\le 1$,
\begin{align*}
&\|S_{n,i\gamma}\|_{\F,\F}\le (1-\|G_{n,i\gamma}\|_{\F,\F})^{-2}(\|A_{n}\|^{2}_{\H_{1/2},\F}\|R_{i\gamma}\|_{\H_{-1/2},\H_{1/2}}+\|E_{n}\|_{\F,\F})\\
\le&(1-\|A_{n}\|_{\H_{1/2},\F}\,|\gamma|^{-1/2})^{-2}(\|A_{n}\|^{2}_{\H_{1/2},\F}\|R_{i\gamma}\|_{\H_{-1/2},\H_{1/2}}+\|E_{n}\|_{\F,\F}) \\
&\le 4 (\|A_{n}\|^{2}_{\H_{1/2},\F}+\|E_{n}\|_{\F,\F})
\end{align*}
and so $\|S_{n,i\gamma} R_{i\gamma}\|_{\F,\F}<1$ whenever $|\gamma|$ is sufficiently large.
Then, by \eqref{schur}, one has
$$
(1-G_{n,i\gamma})R_{n,i\gamma}(1-G_{n,-i\gamma}^{*})=(-(H-S_{n,i\gamma})+i\gamma)^{-1}=
R_{i\gamma}-R_{i\gamma}(1+S_{n,i\gamma}R_{i\gamma})^{-1}S_{n,i\gamma}R_{i\gamma}\,.
$$
Hence,
\begin{align*}
&(-H+i\gamma)-(-H+i\gamma)(1-G_{n,i\gamma})R_{n,i\gamma}(1-G_{n,-i\gamma}^{*})(-H+i\gamma)\\
=&(1+S_{n,i\gamma}R_{i\gamma})^{-1}S_{n,i\gamma}\in\B(\F)\,.
\end{align*}
By $G_ {n,z}\in\B(\H_{1})$, $R_{n,z}\in\B(\F,\H_{1})$ and, by duality, $G^{*}_ {n,z}\in\B(\H_{-1})$, 
$R_{n,z}^{*}=R_{n,\bar z}\in\B(\H_{-1},\F)$. In conclusion,  $(\Thl_{n}+\M_{n,i\gamma})^{-1}\in \B(\F\oplus\F)$.\par
Finally, by \eqref{TS-n}, one obtains, for any $z\in\varrho(H)\cap\varrho(H_{n}-E_{n})$,
\begin{align*}
&R_{z}+\G_{n,z}\big(\Thl_{n}+\Ml_{n,z}\big)^{-1}\G_{n,\bar z}^{*}
=R_{z}+\begin{bmatrix}G_{n,z}&R_{z}\end{bmatrix}\big(\Thl_{n}+\Ml_{n,z}\big)^{-1}
\begin{bmatrix}G^{*}_{n,\bar z}\\R_{z}\end{bmatrix}\\
=&R_{z}+\begin{bmatrix}G_{n,z}&1\end{bmatrix}\begin{bmatrix} R_{n,z}&
R_{n,z}(1-G_{n,\bar z}^{*})\\
(1-G_{n,z})R_{n,z}& -R_{z}+(1-G_{n,z})R_{n,z}(1-G_{n,\bar z}^{*})\big) 
\end{bmatrix}
\begin{bmatrix}G^{*}_{n,\bar z}\\1\end{bmatrix}\\
=&R_{z}+G_{n,z}R_{n,z}G_{n,\bar z}^{*}+G_{n,z}R_{n,z}(1-G_{n,\bar z}^{*})+
(1-G_{n,z})R_{n,z}G_{n,\bar z}^{*} -R_{z}
\\&\ \ \ \, +(1-G_{n,z})R_{n,z}(1-G_{n,\bar z}^{*})\\
=&R_{n,z}=(-(H_{n}-E_{n})+z)^{-1}\,.
\end{align*}
\end{proof}
The resolvent formulae \eqref{krf} and \eqref{Kn} lead to the following
\begin{theorem}\label{teo-conv} Let $A$, $S$, $T_{S}$ and $H_{S}$ be as in Theorem \ref{real}; let  
$A_{n}$ and $H_{n}$ be as in Lemma \ref{krf-n}. Suppose that 
\be\label{convA}
\lim_{n\nearrow\infty}\|A_{n}-A\|_{\H_{1},\F}=0
\ee
and that there exists a sequence $\{E_{n}\}_{1}^{\infty}$ of bounded symmetric operators  in $\F$ such that 
\be\label{unif}
\text{ $(E_{n}-A_{n}RA_{n}^{*})$ is uniformly $(H_{n}-A_{n}RA_{n}^{*})$-small } 
\ee
and
\be\label{ARA*-TS}
\lim_{n\nearrow\infty}\,\sup_{\{\psi\in\S\,:\, \|T_{S}\psi\|^{2}+\|\psi\|^{2}\le 1\}}\|(E_{n}-A_{n}RA_{n}^{*})\psi-T_{S}\psi\|=0\,.
\ee
Then  
\be\label{NR}
\lim_{n\nearrow\infty}(H_{n}-E_{n})=H_{S}\quad\text{in norm resolvent sense.}
\ee
 \end{theorem}
\begin{proof}  
At first, we consider the case $S=0$ and $E_{n}=A_{n}RA_{n}^{*}$ (notice that with such a choice, \eqref{unif} and \eqref{ARA*-TS} trivially hold true). Then, one has
$$
\Thl_{n}=\Thl^{\circ}_{n}:=\begin{bmatrix} 0&1-G_{n}^{*}\\1-G_{n}&-R  
\end{bmatrix}\,,\qquad \Thl_{S}=\Thl_{0}:=\begin{bmatrix} 0&1-G^{*}\\1-G&-R  
\end{bmatrix}\,.
$$
By \eqref{krf} and \eqref{Kn}, one has, for any $z\in\CO\backslash\RE$,
\begin{align*}
&(-(H_{n}-A_{n}RA_{n}^{*})+z)^{-1}-(-H_{0}+z)^{-1}\\
=&\G_{n,z}(\Thl^{\circ}_{n}+\M_{n,z})^{-1}\G_{n,\bar z}^{*}
-\G_{z}(\Thl_{0}+\M_{z})^{-1}\G_{\bar z}^{*}\\
=&\G_{n,z}(\Thl^{\circ}_{n}+\M_{n,z})^{-1}(\G^{*}_{n,\bar z}-\G_{\bar z}^{*})
-(\G_{z}-\G_{n,z})(\Thl_{0}+\M_{z})^{-1}\G_{\bar z}^{*}\\
&+
\G_{n,z}\big((\Thl^{\circ}_{n}+\M_{n,z})^{-1}-(\Thl_{0}+\M_{z})^{-1}\big)\G_{\bar z}^{*}\,.
\end{align*}
By \eqref{convA}, 
\be\label{convGn}
\lim_{n\nearrow\infty}\|\G_{n,z}-\G_{z}\|_{\F\oplus\H_{-1},\F}=
\lim_{n\nearrow\infty}\|\G^{*}_{n,z}-\G_{z}^{*}\|_{\F,\F\oplus\H_{1}}=0
\ee
and so 
\be\label{NR0}
\lim_{n\nearrow\infty}(H_{n}-A_{n}RA_{n}^{*})=H_{0}\quad\text{in norm resolvent sense.}
\ee
amounts to show that for some $z_{\circ}\in\CO\backslash\RE$ there holds
\be\label{convTM0}
\lim_{n\nearrow\infty}\|(\Thl^{\circ}_{n}+\M_{n,z_{\circ}})^{-1}-(\Thl_{0}+\M_{z_{\circ}})^{-1}\|_{\F\oplus\H_{1},\F\oplus\H_{-1}}=0\,.
\ee
By \eqref{convGn} and by the relations
$$
\M_{n,z}-\M_{z}=(z-\lambda_{\circ})(\G_{n}^{*}\G_{n,z}-\G^{*}\G_{z})
=(z-\lambda_{\circ})\big((\G_{n}^{*}(\G_{n,z}-\G_{z})-(\G^{*}-\G_{n}^{*})\G_{z})\big)\,,
$$
one gets, for any $z\in\varrho(H)$, 
\be\label{convMn}
\lim_{n\nearrow\infty}\|\M_{n,z}-\M_{z}\|_{\F\oplus\H_{-1},\F\oplus\H_{1}}=0\,.
\ee
By
$$
\Thl_{0}-\Thl^{\circ}_{n}=\begin{bmatrix} 0&G_{n}^{*}-G^{*}\\G_{n}-G&0  
\end{bmatrix}:\F\oplus\F\to\F\oplus\F\,,
$$
and by
\be\label{GNzGz}
\lim_{n\nearrow\infty}\|G_{n,z}-G_{z}\|_{\F,\F}=
\lim_{n\nearrow\infty}\|G^{*}_{n,z}-G_{z}^{*}\|_{\F,\F}=0\,,
\ee
one gets
$$
\lim_{n\nearrow\infty}\|\Thl_{0}-\Thl^{\circ}_{n}\|_{\F\oplus\F,\F\oplus\F}=0\,.
$$
Hence, for any $z\in\varrho(H)$, 
\be\label{LF}
\lim_{n\nearrow\infty}\|\Thl_{0}+\M_{z}-(\Thl^{\circ}_{n}-\M_{n,z})\|_{\F\oplus\F,\F\oplus\F}=0\,.
\ee
By the continuity of the inversion map in $\{L\in\B(\F\oplus\F):L^{-1}\in\B(\F\oplus\F)\}$ (notice that $(\Thl_{0}+\M_{z})^{-1}\in\B(\F\oplus\F)$ by Remark \ref{RS=0}), that implies, for any $z\in\CO\backslash\RE$, 
\be\label{LFI}
\lim_{n\nearrow\infty}\|(\Thl_{0}+\M_{z_{\circ}})^{-1}-(\Thl^{\circ}_{n}-\M_{n,z_{\circ}})^{-1}\|_{\F\oplus\F,\F\oplus\F}=0\,;
\ee
hence \eqref{convTM0} follows and \eqref{NR0} holds true.
\par
Let us now consider the case $S\not=0$. By \eqref{krf} and \eqref{Kn}, by the same arguments as in the $S=0$ case,
\eqref{NR} amounts to show that for some $z_{\circ}\in\CO\backslash\RE$ there holds
\be\label{convTM}
\lim_{n\nearrow\infty}\|(\Thl_{n}+\M_{n,z_{\circ}})^{-1}-(\Thl_{S}+\M_{z_{\circ}})^{-1}\|_{\F\oplus\H_{1},\F\oplus\H_{-1}}=0\,.
\ee
Considering the invertible operators   
$$
\Thl_{n}+\M_{n,z}:\F\oplus\F\to\F\oplus\F\,,\qquad \Thl_{S}+\M_{z}:\S\oplus\F\to\F\oplus\H_{1}\,,
$$ 
the identity
\begin{align*}
&(\Thl_{n}+\M_{n,z})^{-1}-(\Thl_{S}+\M_{z})^{-1}
=
(\Thl_{n}+\M_{n,z})^{-1}(\Thl_{S}-\Thl_{n}+\M_{z}-\M_{n,z})(\Thl_{S}+\M_{z})^{-1}\\
\end{align*}
holds in $\B(\F\oplus\H_{1},\F\oplus\F)$.
By \eqref{convMn}, by
$$
\Thl_{S}-\Thl_{n}=\begin{bmatrix} T_{S}-(E_{n}-A_{n}RA_{n}^{*})&G_{n}^{*}-G^{*}\\G_{n}-G&0  
\end{bmatrix}:\S\oplus\F\to\F\oplus\F\,,
$$
by \eqref{ARA*-TS} and by \eqref{GNzGz}, in order to get \eqref{NR} it remains to show that for some $z_{\circ}\in\CO\backslash\RE$ there holds
\be\label{supTn}
\sup_{n\ge 1}\|(\Thl_{n}+\M_{n,z_{\circ}})^{-1}\|_{\F\oplus\F,\F\oplus\H_{-1}}<+\infty\,.
\ee
By \eqref{TS-n} and by $(-H+z_{\circ})\in\B(\F,\H_{-1})$, the bound \eqref{supTn} holds whenever
\be\label{11}
\sup_{n\ge 1}\|R_{n,z_{\circ}}\|_{\F,\F}
<+\infty\,,
\ee
\be\label{12}
\sup_{n\ge 1}\|R_{n,z_{\circ}}(1-G^{*}_{n,\bar z_{\circ}})(-H+z_{\circ})\|_{\F,\F}
<+\infty\,,
\ee
\be\label{21}
\sup_{n\ge 1}\|(1-G_{n,z_{\circ}})R_{n,z_{\circ}}\|_{\F,\F}
<+\infty\,,
\ee
\be\label{22}
\sup_{n\ge 1}\|(1-G_{n,z_{\circ}})R_{n,z_{\circ}}(1-G^{*}_{n,\bar z_{\circ}})(-H+z_{\circ})\|_{\F,\F}
<+\infty\,,
\ee
where $R_{n,z}:=(-(H_{n}-E_{n})+z)^{-1}$. Since, by \eqref{GNzGz}, there holds
$$
\sup_{n\ge 1}\|G_{n,z_{\circ}}\|_{\F,\F}<+\infty\,,
$$
it suffices to prove \eqref{11} and \eqref{12}. By \eqref{NR0}, one gets, for any $z\in\CO\backslash\RE$,
\be\label{sup1}
\sup_{n\ge 1}\big\|R_{n,z}^{\circ}\big\|_{\F,\F}<+\infty\,,\qquad 
R_{n,z}^{\circ}:=(-(H_{n}-A_{n}RA_{n}^{*})+z)^{-1}\,.
\ee
Furthermore, by \eqref{TS-n} with $E_{n}=A_{n}RA_{n}^{*}$ and by \eqref{LFI}, one gets, for any $z\in\CO\backslash\RE$,
\be\label{sup2}
\sup_{n\ge 1}\big\|R_{n,z}^{\circ}(1-G_{n,\bar z}^{*})(-H+z)\big\|_{\F,\F}<+\infty\,.
\ee
By writing $H_{n}-E_{n}=H_{n}-A_{n}RA_{n}^{*}-(E_{n}-A_{n}RA_{n}^{*})$, the assumption \eqref{unif} gives, 
for any real $\gamma$ such that $|\gamma|$ is sufficiently large, 
$$\sup_{n\ge 1}\big\|(E_{n}-A_{n}RA_{n}^{*})R_{n,i\gamma}^{\circ}\big\|_{\F,\F}<1 
\,.
$$
That entails
$$
\sup_{n\ge 1}\big\|(1+(E_{n}-A_{n}RA_{n}^{*})R_{n,i\gamma}^{\circ}\big)^{-1}\big\|_{\F,\F}\le
\sup_{n\ge 1}\, \big(1-\big\|(E_{n}-A_{n}RA_{n}^{*})R_{n,i\gamma}^{\circ}\big\|_{\F,\F}\big)^{-1}<+\infty
$$
and
$$
R_{n,i\gamma}=\big(1+(E_{n}-A_{n}RA_{n}^{*})R_{n,i\gamma}^{\circ}\big)^{-1}R_{n,i\gamma}^{\circ}\,.
$$
Therefore, whenever $z_{\circ}= i\gamma$,  one gets, by \eqref{sup1},
\begin{align*}
&\sup_{n\ge 1}\big\|R_{n,z_{\circ}}\big\|_{\F,\F}=\sup_{n\ge 1}\big\|\big(1+(E_{n}-A_{n}RA_{n}^{*})R_{n,z_{\circ}}^{\circ}\big)^{-1}R_{n,z_{\circ}}^{\circ}\big\|_{\F,\F}\\
\le&\sup_{n\ge 1}\big\|\big(1+(E_{n}-A_{n}RA_{n}^{*})R_{n,z_{\circ}}^{\circ}\big)^{-1}\big\|_{\F,\F}\,\|R_{n,z_{\circ}}^{\circ}\big\|_{\F,\F}<
+\infty
\end{align*}
and, by \eqref{sup2}, 
\begin{align*}
&\sup_{n\ge 1}\|R_{n,z_{\circ}}(1-G^{*}_{n,\bar z_{\circ}})(-H+z_{\circ})\|_{\F,\F}\\
=&\sup_{n\ge 1}\|\big(1+(E_{n}-A_{n}RA_{n}^{*})R_{n,z_{\circ}}^{\circ}\big)^{-1}
R_{n,z_{\circ}}^{\circ}(1-G^{*}_{n,\bar z_{\circ}})(-H+z_{\circ})\|_{\F,\F}\\
\le&\sup_{n\ge 1}\|\big(1+(E_{n}-A_{n}RA_{n}^{*})R_{n,z_{\circ}}^{\circ}\big)^{-1}\|_{\F,\F}\,\|
R_{n,z_{\circ}}^{\circ}(1-G^{*}_{n,\bar z_{\circ}})(-H+z_{\circ})\|_{\F,\F}
<+\infty\,;
\end{align*}
hence \eqref{convTM} follows and \eqref{NR} holds true.
\end{proof}
\begin{remark} By the arguments in \cite[Section 3.1]{MPAG} (which build on \cite{LS}, \cite{Sch1}, \cite{Sch2}), Theorems \ref{real} and \ref{teo-conv} applies to the Nelson model. In this case,  with $\F$, $H$ and $A$ as in the introduction, one has  $A\in\B(\H_{s},\F)$, $s>3/4$, and both $\ker(A)$ and $\ran(A)$ are dense in $\F$. Furthermore, in agreement with the renormalization counter term appearing in the introduction, by taking $E_{n}$ equal to the (multiplication operator  by the) leading term in the expansion in the coupling constant $g$ of (minus) the ground state energy at zero total momentum of the regularized Hamiltonian $H_{n}$ corresponding to the ultraviolet cutoff $\Lambda=n$, i.e., by taking $E_{n}=g^{2}N{\mathcal E}_{n}$, where (see \cite[relation (8)]{N})
$$
{\mathcal E}_{n}:=\frac{1}{2(2\pi)^{3}}\int_{|\kappa|< n}(|\kappa|^{2}+\mu^{2})^{-1/2}\left(\frac{|\kappa|^{2}}{2m}+(|\kappa|^{2}+\mu^{2})^{1/2}\right)^{\!\!-1}d\kappa\,,
$$ 
one gets \eqref{ARA*-TS} with a suitable $H$-small symmetric operator $\widehat S$. Therefore, $\widehat H:=H_{\widehat S}$ coincides with the Nelson Hamiltonian. For the precise definition of $T_{\widehat S}$ (and consequently of $\widehat S$) and for an in-depth study of its properties, we refer to \cite[Section 3]{LS}.\par In a similar way, Theorems \ref{real} and \ref{teo-conv} can be applied to other QFT renormalizable models, as the ones studied in  \cite{L2},   \cite{L3}, \cite{LS}, \cite{Sch1}. For recent results regarding the self-adjointness domain and the resolvent of Spin-Boson Hamiltonians with ultraviolet divergences we refer to \cite{LL} and \cite{Lo}.
\end{remark}

\end{section}
\vskip90pt
{\bf Declarations.}
The author has no relevant financial or non-financial interests to disclose.
The author has no competing interests to declare that are relevant to the content of this article.
Data sharing not applicable to this article as no data sets were generated or analyzed during the current study.


\vskip50pt
\begin{thebibliography}{99}


\bibitem{Arai} A. Arai: {\it Analysis on Fock spaces and mathematical theory of quantum fields: An introduction to mathematical analysis of quantum fields.} 
 World Scientific,  Singapore 2018.
 
\bibitem{CFP} C. Cacciapuoti, D. Fermi, A. Posilicano: On inverses of Kre\u\i n's $\mathscr Q$-functions, {\it Rend. Mat. Appl.} {\bf 39} (2018), 229-240.

\bibitem{GW2} M. Griesemer, A. W\"unsch: On the domain of the Nelson Hamiltonian. {\it J. Math. Phys.} {\bf 59} (2018), 042111, 21 pp. 

\bibitem{K} T. Kato: {\it Perturbation Theory for Linear Operators.} Springer 1976.

\bibitem{KP} S. G. Kre\u\i n, Yu. I. Petunin: Scales of Banach Spaces. {\it Russ. Math. Surv.} {\bf 21} (1966), 85-159.

\bibitem{L2} J. Lampart: The Renormalized Bogoliubov-Fr\"ohlich Hamiltonian. {\it  J. Math. Phys.} {\bf 61} (2020), 101902, 12 pp. 

\bibitem{L3} J. Lampart: Hamiltonians for polaron models with subcritical ultraviolet singularities. 
{\it Ann. Henri Poincar\'e} {\bf 24} (2023), 2687-2728.

\bibitem{LS} J. Lampart, J. Schmidt: On Nelson-type Hamiltonians and abstract boundary conditions. {\it Comm. Math. Phys. }{\bf 367} (2019), 629-663. 

\bibitem{LL} S. Lill, D. Lonigro: Self-adjointness and domain of generalized spin-boson models with mild ultraviolet divergences.
{\tt arXiv:2307.14727 [math-ph]}, 2023

\bibitem{Lo} D. Lonigro: Self-Adjointness of a Class of Multi-Spin-Boson Models with Ultraviolet Divergences. {\it Math. Phys. Anal. Geom.} {\bf 26}, 15 (2023)

\bibitem{N0} E. Nelson: Schr\"odinger particles interacting with a quantized scalar field,  in W.T. Martin, I. Segal (eds.) {\it Analysis in Function Space}, 87-120, MIT Press, Cambridge, MA, 1964. 

\bibitem{N} E. Nelson: Interaction of nonrelativistic particles with a quantized scalar field. 
{\it J. Math. Phys.} {\bf 5} (1964), 1190-1197.

\bibitem {JFA} A. Posilicano: A Kre\u{\i}n-like formula for singular
perturbations of self-adjoint operators and applications. \emph{J. Funct.
Anal.}, \textbf{183} (2001), 109-147.

\bibitem {P03}A. Posilicano: Self-adjoint extensions by additive
perturbations. \emph{Ann. Sc. Norm. Super. Pisa Cl. Sci.}(V) \textbf{2} (2003), 1-20.

\bibitem {O&M} A. Posilicano: Self-adjoint extensions of restrictions,
\emph{Oper. Matrices}, \textbf{2} (2008), 483-506.

\bibitem {MPAG} A. Posilicano: On the Self-Adjointness of H+A${\,}^{\!\!*}$+A. {\it Math. Phys. Anal. Geom.} {\bf 23} (2020), 37. 

\bibitem{Sch1} J. Schmidt: On a direct description of pseudorelativistic Nelson Hamiltonians. {\it 
J. Math. Phys.} {\bf 60} (2019), 102303, 21 pp. 

\bibitem{Sch2} J. Schmidt: The Massless Nelson Hamiltonian and its Domain. In: A. Michelangeli (ed.) {\it Mathematical challenges of zero-range physics - models, methods, rigorous results, open problems.} Springer INdAM Series, 42. Springer, Cham, 2021. 

\bibitem{St} M. H. Stone: {\it Linear transformations in Hilbert space.} 
American Mathematical Society. New York, 1932.

\end{thebibliography}
\end{document}